\documentclass[11pt]{amsart}

\usepackage{amsmath}
\usepackage{amsthm}
\usepackage{amsfonts}

\newcommand{\del}{{\partial}}

\theoremstyle{plain}
\newtheorem{prop}{{Proposition}}

\theoremstyle{definition}

\begin{document}
\title[A black hole with no MTT asymptotic to its event horizon]{A black hole with no marginally trapped tube asymptotic to its event horizon}
\author{Catherine Williams} 
\address{Department of Mathematics, Stanford University, Stanford, CA 94305}
\email{cathwill@math.stanford.edu}      
\date{\today}


\begin{abstract}
We construct an example of a spherically symmetric black hole interior in which there is \emph{no} (spherically symmetric) marginally trapped tube asymptotic to the event horizon.  
The construction uses a self-gravitating massive scalar field matter model, and the key condition we impose is that the scalar field $\phi$ be bounded below by a positive constant along the event horizon.
\end{abstract}

\maketitle

\hyphenation{mar-gin-al-ly}


\section{Introduction}

Marginally trapped tubes, those hypersurfaces of spacetime that are foliated by apparent horizons, currently play an important role in mathematical black hole research.
When they exist, these hypersurfaces generally lie inside of black holes and, roughly speaking, form a boundary between the regions of weak and strong gravitational fields there.
Marginally trapped tubes that are everywhere spacelike or null are often referred to as dynamical or isolated horizons in the physics literature, and in some contexts, these are treated as alternate models of black holes' boundaries.
The interested reader is referred to \cite{A} for an introduction and physical motivation, \cite{B, SKB} for numerical results and examples, and \cite{AMSS, AMS, AMS2, AG, CM} for some recent mathematical developments.


The expectation in the physics community is that marginally trapped tubes will generally form during gravitational collapse --- in particular, that any physically reasonable black hole will contain one ---  becoming spacelike or null at late times and asymptotically approaching the event horizon. 
In the special case of spherical symmetry, quite a bit is known about this conjecture in regard to those marginally trapped tubes whose foliating apparent horizons are round, i.e., those tubes which are themselves spherically symmetric.
(Even in a spherically symmetric spacetime, there may exist many distinct \emph{non}-spherically symmetric marginally trapped tubes, but at most one spherically symmetric one.)
In the Schwarzschild and Reissner-Nordstr\"om spacetimes, the spherically symmetric marginally trapped tubes in fact coincide exactly with the black holes' event horizons.
In Vaidya black hole spacetimes where the dominant energy condition is satisfied, the spherically symmetric marginally trapped tube is achronal and either becomes asymptotic to or eventually coincides with the event horizon  \cite{A, Wil}.
This same `nice' marginally trapped tube behavior occurs in black hole spacetimes evolving from sufficiently small spherically symmetric initial data for the Einstein equations coupled with various matter models, namely a scalar field, the Maxwell equations and a real scalar field, the Vlasov equation (describing a collisionless gas), or a Higgs field  \cite{C2}, \cite{D1, DRod}, \cite{DR}, \cite{Wil2}.


The object of this paper, however, is to construct a spherically symmetric black hole interior that provides a counterexample of sorts ---  a black hole which does \emph{not} contain a (spherically symmetric) marginally trapped tube asymptotic to its event horizon. 
This constitutes the first known example of a black hole not exhibiting the expected marginally trapped tube behavior.%
\footnote{Christodoulou has constructed examples of spherically symmetric self-gravitating massless scalar field spacetimes that do not contain (spherically symmetric) marginally trapped tubes, but these spacetimes also do not contain black holes --- rather, they exhibit naked singularities \cite{C3, C4}.}
Our construction uses a massive scalar field matter model, that is, the Einstein equations coupled to the Klein-Gordon equation $\Box_g \phi = \mu \phi$, where $\mu > 0$ is constant.
We specify (spherically symmetric) initial data along two characteristic hypersurfaces in such a way that one of the hypersurfaces could coincide with the event horizon of a black hole.
(This setup is similar to that in \cite{Wil}.)
We then show that the maximal future development of this initial data --- which represents the interior of the black hole --- contains \emph{no} spherically symmetric marginally trapped tube asymptotic to the event horizon; indeed, it need not contain a spherically symmetric marginally trapped tube at all.


Several comments are in order.
First, there are two distinct senses in which one can interpret the phrase `asymptotic to the event horizon' here.
For convenience, we shall call these the causal and geometric senses, respectively; the former is the more commonly accepted notion, while the latter is in some ways more geometrically intuitive. 
(Definitions are given in Section \ref{asymp}.)
Although the two notions are not equivalent --- neither implies the other \emph{a priori} --- in each of the `nice behavior' results described above, the spherically symmetric marginally trapped tube is asymptotic to the event horizon in both senses.
However, for the nonexistence result of this paper, we rule out each case separately: Proposition \ref{causal} constructs an example of a black hole in which no (spherically symmetric) marginally trapped tube can be asymptotic to the event horizon in the causal sense, and Proposition \ref{geometric} then refines that example, adding additional hypotheses, to insure that none is asymptotic to the event horizon in the geometric sense, either.
Since the two examples share most of their properties, henceforth we refer simply to `our construction' when it is not important to distinguish between them.


Second, our construction entails prescribing the initial data in such a way that the scalar field $\phi$ does not decay along the 
characteristic hypersurface representing the event horizon and is in fact bounded below there by a  positive constant.
(The additional hypotheses of Proposition \ref{geometric} require that $\phi$ be bounded above along the event horizon as well.)
However, while nothing in the literature currently rules out such a possibility, it is not clear whether such non-decaying data could actually arise on the event horizon of a black hole which had evolved from asymptotically flat initial data.
%
%
This construction should therefore be thought of as exhibiting a mechanism by which the marginally trapped tube can be pushed away from the event horizon inside a black hole, rather than as a fully formed counterexample to the conjecture that a marginally trapped tube arising in gravitational collapse must be asymptotic to the event horizon.
(For the latter, one would need to specify asymptotically flat initial data for the given Einstein-matter system, show that the maximal future development of that initial data contains a black hole region, and \emph{then} prove that any marginally trapped tube lying in the black hole does not become asymptotic to the event horizon.
Concocting such an example would require a delicate balancing act: the initial data would have to be sufficiently small and/or decay sufficiently rapidly to insure that future null infinity forms ($\mathcal{I}^+ \neq \emptyset$), but sufficiently large that enough matter falls into the black hole to push any marginally trapped tube away from the event horizon as in our construction.)


Third, we note that general spherically symmetric black hole interiors were considered in \cite{Wil}, and it was shown there that if four particular inequalities involving the metric and stress-energy tensor are satisfied near future timelike infinity $i^+$, the future limit point of the event horizon, then the black hole \emph{must} contain a marginally trapped tube that asymptotically approaches the event horizon.
Our result here is fully consistent with that theorem, since two of the hypotheses of the latter --- inequalities A and B2 --- are immediately violated by the way we specify our scalar field initial data.


Finally, we emphasize that here we prove only that no \emph{spherically symmetric} marginally trapped tube is asymptotic to the constructed black hole's event horizon; it remains possible that there exists some non-spherically symmetric one in our constructed black hole which is.
The asymptotic behavior of non-spherically symmetric marginally trapped tubes is a subtle issue, and very little is currently known other than Corollary 4.6 of \cite{AG}.


\section{Background assumptions}\label{ba}

\subsection{Self-gravitating massive scalar fields \& spherical symmetry}

A self-gravitating massive scalar field consists of a 4-dimensional spacetime $(\mathcal{M},g)$ and a scalar function $\phi \in C^2(\mathcal{M})$ satisfying the coupled Einstein-Klein-Gordon system:
\begin{equation}\label{eeq} R_{\alpha \beta} - \textstyle \frac{1}{2} R g_{\alpha \beta} = 2 T_{\alpha \beta} \end{equation} 
\begin{equation}\label{sfT} T_{\alpha \beta} = \phi_{; \alpha} \phi_{; \beta} - \left( \textstyle\frac{1}{2} \phi_{;\gamma} \phi^{;\gamma} + \frac 12 \mu \phi^2 \right) g_{\alpha \beta} \end{equation}
\begin{equation}\label{sfeq} \Box_g \phi = g^{\alpha \beta}\phi_{;\alpha \beta}  = \mu \phi, \end{equation}
where the constant $\mu > 0$ is the mass of the scalar field.


In general, a spacetime $(\mathcal{M},g)$ is said to be spherically symmetric if the Lie group $SO(3)$ acts on it by isometries with orbits which are either fixed points or spacelike 2-spheres.  
One typically assumes further that the quotient $\mathcal{Q} = \mathcal{M}/SO(3)$ inherits a 1+1-dimensional Lorentzian manifold structure with (possibly empty) boundary.
Suppressing pullback notation, the upstairs metric then takes the form 
\[ g = \overline{g} + r^2 \gamma, \]
where $\overline{g}$ is the Lorentzian metric on $\mathcal{Q}$, $\gamma$ is the usual round metric on $S^2$, and $r$ is a smooth nonnegative function on $\mathcal{Q}$ called the area-radius,  whose value at each point is proportional to the square root of the area of the corresponding two-sphere upstairs in $\mathcal{M}$.  
In a spherically symmetric self-gravitating massive scalar field spacetime, the function $\phi$ must be invariant under the $SO(3)$-action and must therefore also descend to a function on $\mathcal{Q}$.


If in addition the topology of $\mathcal{Q}$ is such that it admits a conformal embedding into a subset of 2-dimensional Minkowski space 
$\mathbb{M}^{1+1}$, one may identify $\mathcal{Q}$ with its image under this embedding and make use of the usual global double-null coordinates $u$, $v$ on $\mathbb{M}^{1+1}$.
The conformal metric $\overline{g}$ may be written
\[ \overline{g} = - \Omega^2 du\, dv, \]
where $\Omega = \Omega(u,v) > 0$ on $\mathcal{Q}$, and 
$\phi = \phi(u,v)$ and $r = r(u,v)$.
Then \eqref{eeq}-\eqref{sfeq} becomes the following system of equations on $\mathcal{Q}$:
\begin{eqnarray}  
\del_u(\Omega^{-2} \del_u r) & = & - r \Omega^{-2} (\del_u\phi)^2  \label{eq1}\\
\del_v(\Omega^{-2} \del_v r) & = & - r \Omega^{-2} (\del_v\phi)^2 \label{eq2}\\
\del_u m & = & \textstyle \frac 12 r^2 \left(  \mu \phi^2 \del_u r - 4 \Omega^{-2} (\del_u\phi)^2 \del_v r \right) \label{eq3}\\
\del_v m & = & \textstyle \frac 12 r^2 \left( \mu \phi^2 \del_v r - 4 \Omega^{-2} (\del_v\phi)^2 \del_u r \right) \label{eq4} \\ 
\mu \phi  & = & -4\Omega^{-2} \left( \del^2_{uv}\phi + \del_u\phi\,(\del_v \log r) + \del_v\phi\,(\del_u \log r) \right) \label{sfequv} 
\end{eqnarray}
where 
\begin{equation} m = m(u,v) = \frac{r}{2}(1 + 4\Omega^{-2} \del_u r \del_v r) \label{mdef} \end{equation}
is the Hawking mass.
The null constraints (\ref{eq1}) and (\ref{eq2}) are Raychaudhuri's equation applied to each of the two null directions in $\mathcal{Q}$.


\subsection{Initial value problem}\label{ivp}

For our construction, we proceed in the opposite direction.
That is, we find conformal, radial, and scalar functions $\Omega$, $r$, and $\phi$, respectively, that satisfy \eqref{eq1}-\eqref{sfequv} on a set $\mathcal{Q} \subset \mathbb{M}^{1+1}$  and whose properties align with those necessary for $(\mathcal{Q} \times S^2, -\Omega^2 du \, dv + r^2 \gamma)$ to be the interior of a black hole in a self-gravitating massive scalar field spacetime.
Our main result then identifies certain additional conditions on the functions $\Omega$, $r$, and $\phi$ that preclude the possibility of a (spherically symmetric) marginally trapped tube asymptotically approaching the event horizon of this black hole.


To begin, let us declare that for any  values $u, v > 0$, $K(u,v)$ denotes the characteristic rectangle given by
\[ K(u,v) = [0, u] \times [v, \infty) \subset \mathbb{M}^{1+1}. \]
Now choose some values $u_0, v_0 > 0$, \emph{fix} the specific rectangle $K(u_0, v_0)$, and define characteristic initial hypersurfaces $\mathcal{C}_{in} := [0, u_0] \times \{ v_0 \}$ and $\mathcal{C}_{out} := \{ 0 \} \times [v_0, \infty)$.
Constant-$v$ curves are interpreted as ingoing, constant-$u$ curves as outgoing.


Fixing a mass parameter $\mu > 0$,  assume we have initial data for $r$, $\Omega$, $\phi$ along $\mathcal{C}_{in} \cup \mathcal{C}_{out}$ and that the data satisfy the following: 
\begin{equation} \del_u r < 0 \,\text{ along }\, \mathcal{C}_{out}, \label{drdusign}\end{equation}
\begin{equation} \del_v r > 0 \,\text{ along }\, \mathcal{C}_{out},\label{drdvsign}\end{equation}
and
\begin{equation} 0 < r < r_+ \,\text{ along }\, \mathcal{C}_{out}, \label{rbounds}\end{equation}
for some constant $r_+ < \infty$.  
Given inequality \eqref{drdvsign}, we further assume that the value $r_+ = \lim_{v\rightarrow \infty} r(0,v)$.
These assumptions are necessary to guarantee that the outgoing initial hypersurface $\mathcal{C}_{out}$ could agree with the (quotient of) the event horizon in a black hole spacetime, i.e.\ $\mathcal{C}_{out} \equiv \mathcal{H}$.
%
See Section 2.2 of \cite{Wil} for an explanation of the significance of each inequality.


Finally let $\mathcal{Q}$ be the maximal future development of this initial data in $K(u_0, v_0) \subset \mathbb{M}^{1+1}$.
We define three subsets of interest in $\mathcal{Q}$: 
the  regular region
\[ \mathcal{R} = \{ (u,v) \in \mathcal{Q} : \del_v r > 0 \text{ and } \del_u r < 0 \}, \] the  trapped region
\[ \mathcal{T} = \{ (u,v) \in \mathcal{Q} : \del_v r < 0 \text{ and } \del_u r < 0 \}, \] and the  marginally trapped tube 
\[ \mathcal{A} = \{ (u,v) \in \mathcal{Q} : \del_v r = 0 \text{ and } \del_u r < 0 \}. \] 
This marginally trapped tube definition agrees with the usual one (e.g., in \cite{A}), since in spherical symmetry the inner and outer expansions $\theta^-$ and $\theta^+$ are proportional to $\del_u r$ and $\del_v r$, respectively.
Note that assumption \eqref{drdvsign} implies that $\mathcal{A} \cap \mathcal{C}_{out} = \emptyset$.
We shall also make use of the fact that the quantity $1- \frac{2m}{r}$ is positive, negative, and zero in $\mathcal{R}$, $\mathcal{T}$, and $\mathcal{A}$, respectively; this follows immediately from \eqref{mdef}.


\subsection{Causal \& geometric asymptotic behavior} \label{asymp}
When one asks whether a marginally trapped tube is asymptotic to the event horizon, it is not immediately obvious what is meant --- the phrase  `asymptotic to the event horizon' can be interpreted in (at least) two different ways in the context of spherical symmetry.
On the one hand, if $\mathcal{A}$ is asymptotic to $\mathcal{H}$, then we expect all inextendible future-directed causal curves intersecting $\mathcal{H}$ at late times to intersect $\mathcal{A}$ also.
Thus, roughly speaking, the marginally trapped tube captures all the same light and matter as the event horizon (at least at late times).
This is the intuition behind what we call the causal sense of `asymptotic to the event horizon.'
On the other hand, since the areas of the two-spheres foliating the event horizon tend to some finite limiting value at late times, we also expect that the areas of the marginally trapped spheres (apparent horizons) foliating the marginally trapped tube should tend to that same limiting value.
Another way of saying this is that the radius of the marginally trapped tube should approach that of the event horizon in the limit.
What we call the geometric sense of `asymptotic to the event horizon' follows this interpretation.


More precisely, suppose we have a black hole interior $\mathcal{Q}$ as constructed above with event horizon $\mathcal{H} = \mathcal{C}_{out} = \{ 0 \} \times [v_0, \infty)$.
We say that the marginally trapped tube $\mathcal{A}$ is asymptotic to the event horizon $\mathcal{H}$ in the causal sense if $\mathcal{A} \cap \{ (u,v) \in \mathcal{Q} : v \geq V \}  = \{ ( f(v), v) : v \geq V \}$ for some function $f \in C([V, \infty))$, $f \geq 0$, such that $f(v) \rightarrow 0$ as $v\rightarrow \infty$, some $V \geq v_0$.
In other words, the curve described by $\mathcal{A}$ in $\mathbb{M}^{1+1}$ is truly asymptotic to the ray described by $\mathcal{H}$ as $v \rightarrow \infty$.
This is the definition of `asymptotic to the event horizon' used in all of the existing `nice behavior' results cited in the introduction.
Furthermore, this definition agrees with the more general one given in \cite{AG} (p.\ 16), which requires that the past Cauchy horizon of the marginally trapped tube coincide with the event horizon, at least in the portion of the spacetime to the future of some achronal spacelike hypersurface which intersects both of them.


For the precise definition of the geometric sense of `asymptotic,' recall that in our setup the black hole has asymptotic area-radius $r_+$ --- that is, $r \nearrow r_+$ along $\mathcal{H} = \mathcal{C}_{out}$.
We say that the marginally trapped tube $\mathcal{A}$ is asymptotic to the event horizon $\mathcal{H}$ in the geometric sense if, with the same setup as above, $\mathcal{A} \cap \{ (u,v) \in \mathcal{Q} : v \geq V \}  = \{ ( f(v), v) : v \geq V \}$ for a function $f \in C([V, \infty))$, $f\geq 0$, such that $r( (f(v), v) ) \rightarrow r_+$ as $v\rightarrow \infty$.
That is, $r \rightarrow r_+$ along $\mathcal{A}$ toward the asymptotic end.
Note that a marginally trapped tube being asymptotic to the event horizon in one sense does not imply that it is asymptotic in the other sense as well, at least not \emph{a priori}.
(In practice, however, the estimates used to show that a marginally trapped tube satisfies one definition often imply immediately that it satisfies the other.)


\section{The Main Results}\label{main}


\begin{prop} \label{causal}
With the setup as described in Section \ref{ivp}, set 
\[ \eta = r\phi. \]
Suppose that 
\begin{equation} \del_u \eta (0, v_0) > 0, \label{corner+} \end{equation}
and that along $\mathcal{C}_{out}$,
\begin{equation} \eta   >  \textstyle\frac{1}{\sqrt{\mu}} + r_+ . \label{etacond} \end{equation}
Then there exists $0 < u_1 \leq u_0$ such that $\mathcal{Q}$ contains a rectangle $K(u_1, v_0)$ in which $\del_v r > 0$.
Thus no (spherically symmetric) marginally trapped tube is asymptotic to the event horizon $\mathcal{H} = \mathcal{C}_{out}$ in the causal sense.
\end{prop}


\noindent
\emph{Remarks.}  
If instead of \eqref{etacond} we had assumed that $|\phi| + |\del_v \phi| = O(v^{-p})$ along $\mathcal{C}_{out}$, some $p > \frac 12$, then with a few additional minor technical assumptions imposed, we could conclude from Theorem 3 of \cite{Wil} that $\mathcal{Q}$ \emph{must} contain a marginally trapped tube asymptotic to the event horizon, in both the causal and geometric senses.
(The potential function $V(\phi)$ used there is $\frac 12 \mu \phi^2$ for a massive scalar field.)
Thus in some sense the non-decay of $\eta$ (and hence of $\phi$) is necessary for this nonexistence result.
More generally, as noted in the introduction, the lower bound \eqref{etacond} for $\eta$ immediately violates conditions A (and A$^\prime$) and B2 of \cite{Wil}, Theorem 1.


In general, $\mathcal{Q}$ need not contain a whole rectangle $K(u, v_0)$ for \emph{any} $u > 0$, and the fact that it does here implies that the spacetime extends all the way out to a Cauchy horizon, $[0, u_1] \times \{ \infty \}$.
Furthermore, since $\del_v r > 0$ in all of $K(u_1, v_0)$, this Cauchy horizon is accessible from the regular region $\mathcal{R}$, an unusual situation (cf.\ \cite{D1}).


Note that this result does not rule out the possibility of a marginally trapped tube asymptotic to the event horizon in the geometric sense.
It is possible, for example, that $r(u_1, v) \nearrow r_+$ as $v \rightarrow \infty$; then $\mathcal{A}$ could lie in $J^+ \left(\{ u_1 \} \times [v_0, \infty) \right) \cap \mathcal{Q}$, asymptotically approach the ray $\{ u_1 \} \times [v_0, \infty)$ (as a curve in $\mathbb{M}^{1+1}$), and have the property that $r \rightarrow r_+$ as $v \rightarrow \infty$ --- this would imply that $\mathcal{A}$ was indeed asymptotic to $\mathcal{H}$ in the geometric sense.
Our second result rules out this scenario, at the expense of imposing additional assumptions on the initial data.


\begin{prop}\label{geometric}
Consider data as in Proposition \ref{causal}.
If we further require that along $\mathcal{C}_{out}$
\begin{equation} m > 0 \label{mpos}, \end{equation}
\begin{equation} \del_v \phi \leq 0, \label{dvphineg}\end{equation}
and 
\begin{equation} \left( \del_u r + r_+ \phi_+ \del_u \phi \right) (0, v_0) < 0 \label{ccond} \end{equation}
for some  constant $\phi_+ > \phi(0,v_0)$, then there exist $0 < u_2 \leq u_1$ and $\delta > 0$ such that $r \leq r_+ - \delta$ everywhere in $J^+(\{ u_2 \} \times [v_0, \infty)) \cap \mathcal{Q}$.
Then since $\mathcal{A} \cap K(u_2, v_0) = \emptyset$ by Proposition \ref{causal}, no (spherically symmetric) marginally trapped tube in $\mathcal{Q}$ is asymptotic to the event horizon in the geometric sense.
\end{prop}


\section{Proofs of the Main Results}\label{main}


\begin{proof}[Proof of Proposition \ref{causal}]
First observe that since $\del_u r < 0$ along $\mathcal{C}_{out}$ by assumption \eqref{drdusign}, equation \eqref{eq1} immediately implies that
\begin{equation} \del_u r < 0  \,\text{ in }\, \mathcal{Q}. \label{drduineq} \end{equation}
Then \eqref{rbounds} and \eqref{drduineq} together imply that 
\begin{equation} r < r_+ \,\text{ in }\, \mathcal{Q}. \label{r+unif}\end{equation}


Choose $0 < u_1 \leq u_0$ sufficiently small that $r$, $\del_u \eta$, and $\del_v r$ are all strictly positive on 
$[0, u_1] \times \{ v_0\}$; such a choice is possible by assumptions \eqref{drdvsign}, \eqref{rbounds}, and \eqref{corner+} by continuity.
Set $K_1 = K(u_1, v_0) = [0, u_1] \times [v_0, \infty)$; henceforth we restrict our attention to $\mathcal{Q} \cap K_1$.


From \eqref{eq2} it follows that for any $(u,v) \in \mathcal{R}$, the outgoing null segment $\{ u \} \times [v_0, v] \subset \mathcal{R}$ as well. 
Setting $r_0 = r(u_1, v_0)$ and applying \eqref{drduineq}, we therefore have
\begin{equation}  0 < r_0 \leq r(u,v)  \label{r0}\end{equation}
for all  $(u,v) \in \mathcal{R} \cap K_1$.


Before proceeding, we clarify our notation with respect to the causal and topological structures in play. 
First, given a point $p$ in $\mathcal{Q}$ or $K_1$, we use $J^-(p)$ to denote its causal past in $(\mathbb{M}^{1+1}, - du\, dv)$, rather than in $(\mathcal{Q}, -\Omega^2 du\, dv)$.
Thus for any $p \in \mathcal{Q}$, $J^-(p)$ is an infinite backwards cone in $\mathbb{M}^{1+1}$; to express its causal past in $\mathcal{Q}$, we write $J^-(p) \cap \mathcal{Q}$.
The topology on $\mathcal{Q}$ is the relative one with respect to $K_1$, which of course inherits its topology from $\mathbb{M}^{1+1}$.
Unless otherwise specified, however, set boundaries and closures are taken with respect to $K_1$ rather than $\mathcal{Q}$.
It is also perhaps worth noting that for $p \in \mathcal{Q} \cap K_1$, $J^-(p) \cap \mathcal{Q} \equiv J^-(p) \cap K_1$, but if $p \in \overline{\mathcal{Q}} \cap K_1$, then $J^-(p) \cap K_1$ may contain points which  $J^-(p) \cap \mathcal{Q}$ does not (namely, points in $\overline{\mathcal{Q}}\setminus\mathcal{Q}$).


Define a region $\mathcal{U} \subset \mathcal{Q} \cap K_1$ to be the set of all points $(u,v)$ such that the following two inequalities are satisfied for all $(\tilde{u}, \tilde{v}) \in J^{-}(u,v) \cap K_1$:
\begin{eqnarray}
\eta(\tilde{u}, \tilde{v}) & > & \frac{1}{\sqrt{\mu}} + r_+ \label{bs1}\\
\del_u \eta(\tilde{u}, \tilde{v}) & > & 0. \label{bs2}
\end{eqnarray}


The proof proceeds in three steps.
We first observe that $\del_v r > 0$ in $\overline{\mathcal{U}} \cap \mathcal{Q}$, i.e.\ $\overline{\mathcal{U}} \cap \mathcal{Q} \subset \mathcal{R}$.
%
%
Next we show that $\mathcal{U}$ fills $\mathcal{Q} \cap K_1$, i.e.\ $\mathcal{U} = \mathcal{Q} \cap K_1$, by showing that $\mathcal{U}$ is both open and closed in $\mathcal{Q} \cap K_1$.
And finally we show that $\left( \overline{\mathcal{Q}}\setminus \mathcal{Q} \right) \cap K_1 = \emptyset$, which implies that $K_1 = \mathcal{Q} \cap K_1 = \mathcal{U}$.
Since $\mathcal{U} \subset \mathcal{R}$, this last statement in turn implies that $\mathcal{A} \cap K_1 = \emptyset$, proving the proposition.


First, to see that $\del_v r > 0$ in $\overline{\mathcal{U}} \cap \mathcal{Q}$, note that combining equations \eqref{eq1}, \eqref{eq3} and \eqref{mdef} (or alternately,  \eqref{eq2}, \eqref{eq4} and \eqref{mdef}) yields
\begin{equation}  \del^2_{uv} r  =  
\textstyle \frac 14 \Omega^2 r^{-1}  \left[ (1 - \frac{2m}{r}) + \mu \eta^2 -1 \right]. \label{duvr}\end{equation}
Consequently by \eqref{bs1}, we have
\begin{eqnarray*} \del^2_{uv} r 
& > & \textstyle \frac 14 \Omega^2 r^{-1}  \left[ (1 - \frac{2m}{r}) + \mu (\frac{1}{\mu} + r_+^2) - 1 \right], \\
& = & \textstyle \frac 14 \Omega^2 r^{-1}  \left[ (1 - \frac{2m}{r}) + \mu r_+^2  \right] \\
& > & 0
\end{eqnarray*}
in $\overline{\mathcal{R}} \cap \overline{\mathcal{U}} \cap \mathcal{Q}$.
Now, suppose $\del_v r(u_\ast,v_\ast) \leq 0$ for some  $(u_\ast,v_\ast) \in \overline{\mathcal{U}} \cap \mathcal{Q}$.
Then since $\del_v r(0, v_\ast) > 0$ by \eqref{drdvsign}, there exists $0 < u_{\ast\ast} \leq u_\ast$ such that $(u_{\ast\ast}, v_\ast)$ is the first point along the segment $[0, u_\ast] \times \{ v_\ast \}$ to leave $\mathcal{R}$ --- that is, $\del_v r (u_{\ast\ast}, v_\ast) = 0$, while $\del_v r (u, v_\ast) > 0$ for all $0 \leq u < u_{\ast\ast}$.
The segment $[0, u_\ast] \times \{ v_\ast \} \subset \overline{\mathcal{U}} \cap \mathcal{Q}$ by definition of $\mathcal{U}$, so  $[0, u_{\ast\ast}] \times \{ v_\ast \} \subset \overline{\mathcal{R}} \cap \overline{\mathcal{U}} \cap \mathcal{Q}$.
But then $\del^2_{uv} r > 0$ along $[0, u_{\ast\ast}] \times \{ v_\ast \}$ while $\del_v r (0,v_\ast) > 0 = \del_v r(u_{\ast\ast}, v_\ast)$, a contradiction.
So in fact $\overline{\mathcal{U}} \cap \mathcal{Q} \subset \mathcal{R}$.


Next, that $\mathcal{U}$ is open in $\mathcal{Q} \cap K_1$ follows immediately from the fact that the inequalities defining $\mathcal{U}$ are strict.
To show that $\mathcal{U}$ is also closed in $\mathcal{Q} \cap K_1$, i.e.\ that $\overline{\mathcal{U}} \cap \mathcal{Q} \subset \mathcal{U}$, we employ a simple bootstrap argument.
Suppose $(u_\ast,v_\ast) \in \overline{\mathcal{U}} \cap \mathcal{Q}$.
Then $\del_u \eta (u, v_\ast) \geq 0$ for all $0 \leq u \leq u_\ast$, which immediately yields
\[ \eta(u_\ast, v_\ast) \geq \eta(0, v_\ast) > \frac{1}{\sqrt{\mu}} + r_+,\]
retrieving inequality \eqref{bs1} at $(u_\ast, v_\ast)$.
To retrieve \eqref{bs2}, note that by combining \eqref{sfequv} and \eqref{duvr} we have
\[ \del^2_{uv} \eta 
= \textstyle \frac 14 \Omega^2 r^{-2} \left[  (1 - \textstyle  \frac{2m}{r})  + \mu (\eta^2 - r^2)  - 1    \right]  \eta . \]
Since $J^-(u_\ast, v_\ast) \cap K_1 \subset \overline{\mathcal{U}} \cap \mathcal{Q} \subset \mathcal{R}$, we have $(1 - \frac{2m}{r}) > 0$ in $J^-(u_\ast, v_\ast) \cap K_1$.
From inequality \eqref{bs1}, we also have $\eta \geq \frac{1}{\sqrt{\mu}} + r_+$ in $J^-(u_\ast, v_\ast) \cap K_1$.
Therefore, at any point in $J^-(u_\ast, v_\ast) \cap K_1$, we have
\begin{eqnarray*} \del^2_{uv} \eta 
& > & \textstyle \frac 14 \Omega^2 r^{-2} \left[  \mu \left( (\frac{1}{\sqrt{\mu}} + r_+)^2 - r^2 \right)  - 1    \right]  \eta \\
& > & \textstyle \frac 14 \Omega^2 r^{-2} \left[  \mu \left( \frac{1}{\mu} + r_+^2 - r^2 \right)  - 1    \right]  \eta \\
& > & 0,
\end{eqnarray*}
since $r < r_+$ everywhere in $\mathcal{Q}$ by \eqref{r+unif}.
In particular, we have $\del^2_{uv} \eta (u_\ast,v) > 0$ for all $v_0 \leq v \leq v_\ast$, 
which implies
\[ \del_u \eta (u_\ast, v_\ast) \geq \del_u \eta(u_\ast, v_0) > 0, \]
since $\del_u \eta > 0$ on $[0, u_1] \times \{ v_0 \}$.
We have thus retrieved strict inequalities \eqref{bs1} and \eqref{bs2} in $\overline{\mathcal{U}} \cap \mathcal{Q}$, from which it follows that $\mathcal{U}$ is closed in $\mathcal{Q} \cap K_1$.


Since $\mathcal{U} \neq \emptyset$ (it at least contains a neighborhood of the point $(0,v_0)$) and $\mathcal{Q} \cap K_1$ is necessarily connected, we therefore have $\mathcal{U} = \mathcal{Q} \cap K_1$.


Finally, to see that $K_1 \subset \mathcal{Q}$, we employ an extension principle known to hold for our matter model, Theorem 3.1 of \cite{D3}, which says that `first' singularities (in the causal sense) can  arise only from the trapped region $\mathcal{T}$ or from the center of symmetry $\Gamma$,  i.e.,  $\Gamma = \{ p \in \mathcal{Q} : r(p) = 0 \}$.
In our setting, we apply this principle as follows:
Suppose that $K_1$ is not contained in $\mathcal{Q}$. 
Then in particular $\left( \overline{\mathcal{Q}} \setminus \mathcal{Q} \right) \cap K_1 \neq \emptyset$, and since $\overline{\mathcal{Q}} \setminus \mathcal{Q} $ is achronal, we can find a point $p_\ast \in \left( \overline{\mathcal{Q}} \setminus \mathcal{Q} \right) \cap K_1 $ such that all  points to the causal past of $p_\ast$ in $K_1$ lie in the spacetime $\mathcal{Q}$, rather than on its future boundary --- that is, such that  $\left( J^-(p_\ast)\setminus \{p_\ast \} \right) \cap K_1 \subset \mathcal{Q} \cap K_1$.
Since $\mathcal{Q} \cap K_1 = \mathcal{U} \subset \mathcal{R}$, we thus have $\left( J^-(p_\ast)\setminus \{p_\ast\} \right) \cap K_1 \subset \mathcal{R}$.
Theorem 3.1 of \cite{D3} then asserts that $p_\ast \in \overline{\Gamma}\setminus \Gamma$, but \eqref{r0} implies that $r \geq r_0 > 0$ everywhere in $\mathcal{Q} \cap K_1$,  a contradiction.
So in fact $\left( \overline{\mathcal{Q}} \setminus \mathcal{Q} \right) \cap K_1 = \emptyset$, which implies that $K_1 = \mathcal{Q} \cap K_1$.
Thus $K_1 = \mathcal{U} \subset \mathcal{R}$, and hence $\mathcal{A} \cap K_1 = \emptyset$, as claimed.
\end{proof}


\begin{proof}[Proof of Proposition \ref{geometric}]
First of all, it follows from assumptions \eqref{mpos} and \eqref{ccond} by continuity that we may choose $0 < u_2 \leq u_1$ sufficiently small that $m \geq 0$, $\phi \leq \phi_+$, and $\del_u r + r_+ \phi_+ \del_u \phi < 0$ along $[0, u_2] \times \{ v_0 \}$, where $u_1$ is as in the proof of Proposition \ref{causal}.
Set $K_2 = K(u_2, v_0) = [0, u_2] \times [v_0, \infty)$.


Note that since all of the hypotheses of Proposition \ref{causal} are still in place, we know from its proof that $\del_u r < 0$, $\del_v r > 0$,  $\eta > 0$, $\del_u \eta > 0$, and $\del^2_{uv} r > 0$ hold everywhere in $K_2 \subset K_1 = \mathcal{Q} \cap K_1$.


Since $\del_u r < 0$ and $\del_v r > 0$ in $K_2$ and $m \geq 0$ along $[0, u_2] \times \{ v_0 \}$, it follows immediately from equation \eqref{eq4} that $m \geq 0$ and hence $1 - \frac{2m}{r} \leq 1$ everywhere in $K_2$.
Then recalling equation \eqref{duvr}, we have
\begin{eqnarray} \del^2_{uv} r  
& = & \textstyle \frac 14 \Omega^2 r^{-1}  \left[ (1 - \frac{2m}{r}) + \mu \eta^2 -1 \right] \nonumber \\
& \leq & \textstyle \frac 14 \mu \Omega^2 r^{-1}   \eta^2 \label{duvrest}
\end{eqnarray}
in $K_2$.
Combining equations \eqref{sfequv} and \eqref{duvr} yields 
\[ \textstyle\frac{1}{4} \mu \Omega^{2} r^{-1} \phi  = - r^{-1} \left( \del^2_{uv}\phi + \del_u\phi\,(\del_v \log r) + \del_v\phi\,(\del_u \log r) \right), \]
and substituting this into \eqref{duvrest}, we find that
\begin{eqnarray} \del^2_{uv} r  
& \leq &  - r^{-1} \left( \del^2_{uv}\phi + \del_u\phi\,(\del_v \log r) + \del_v\phi\,(\del_u \log r) \right) \cdot r^2 \phi  \nonumber \\
& = &  - \eta \left( \del^2_{uv}\phi + \del_u\phi\,(\del_v \log r) + \del_v\phi\,(\del_u \log r) \right).  \label{duvrest1.5} 
\end{eqnarray}
everywhere in $K_2$. 


Recall that $\del_u \eta > 0$ in $K_2$.
This implies that
\[ 0 < r \del_u \phi + \phi \del_u r, \]
so since $\phi > 0$ (since $\eta > 0$), $r > 0$, and $\del_u r < 0$, we must have $\del_u \phi > 0$ in $K_2$.
Thus from \eqref{duvrest1.5}  we obtain
\begin{equation} \del^2_{uv} r  <  - \eta \left( \del^2_{uv}\phi  + \del_v\phi\,(\del_u \log r) \right). \label{duvrest2}
\end{equation}


Now, the righthand side of \eqref{duvrest2} must be positive in $K_2$ because $\del^2_{uv} r$ is, which means that
\begin{equation}\label{duvphi} \del^2_{uv}\phi  + \del_v\phi\,(\del_u \log r) < 0 \end{equation}
everywhere in $K_2$.
Setting $\mathcal{V} = \{ p \in K_2 : \del_v \phi(p) \leq 0 \}$, inequality \eqref{duvphi} implies that $\del^2_{uv} \phi < 0$ in $\mathcal{V}$--- thus $\del_v \phi$ must decrease along ingoing null rays in $\mathcal{V}$.
It then follows from \eqref{dvphineg} that $K_2 \subset \mathcal{V}$.
That is, $\del_v \phi \leq 0$, and hence $\del^2_{uv} \phi < 0$,  in all of $K_2$.


Since $\phi(u_2, v_0) \leq \phi_+$ by our choice of $u_2$, inequalities $\del_u \phi > 0$ and $\del_v \phi \leq 0$ together imply that $\phi(p) \leq \phi_+$ for all $p \in K_2$.
Thus, recalling \eqref{r+unif}, we have
\begin{equation*} \eta < r_+ \phi_+ \,\text{ in } \, K_2.
\end{equation*}


Finally, returning to \eqref{duvrest2} we have
\begin{eqnarray} 
\del^2_{uv} r  
& < &  - \eta \left( \del^2_{uv}\phi  + \del_v\phi\,(\del_u \log r) \right) \nonumber \\
& \leq &  - \eta  \del^2_{uv}\phi \nonumber \\
& < &  - r_+ \phi_+ \del^2_{uv}\phi. \label{duvrest3}
\end{eqnarray}
Integrating \eqref{duvrest3} along an outgoing null ray $\{ u \} \times [v_0, v] \subset K_2$, we have
\[ \del_u r(u,v) - \del_u r(u, v_0) <   - r_+ \phi_+  (\del_u \phi(u,v) - \del_u \phi(u, v_0)).
\]
Setting $\delta(u) = - \left( \del_u r(u, v_0) + r_+ \phi_+ \del_u \phi(u, v_0) \right)$ and using the fact that \linebreak $\del_u \phi(u,v) > 0$, we then obtain
\begin{equation} \del_u r(u,v)  <  -\delta(u) \leq - \delta_0, \label{drduest}
\end{equation}
where $\delta_0 = \inf_{0 \leq u \leq u_2} \delta(u)$; our choice of $u_2$ guarantees that $\delta_0 > 0$.
Finally, integrating \eqref{drduest} along an ingoing null segment $[0, u] \times \{ v\}$, we have
\begin{eqnarray*} 
r(u,v) - r(0,v)  & < &  - \delta_0\cdot u,
\end{eqnarray*}
so in particular,
\[ r(u_2, v) < r_+ - \delta_0 \cdot u_2 \]
for all $v \geq v_0$.
Setting $\delta = \delta_0 \cdot u_2$, this implies the result of the proposition, since $\del_u r < 0$ in all of $\mathcal{Q}$.
\end{proof}


\section*{Acknowledgments}

The author would like to thank Pieter Blue for his assistance interpreting some of the physics literature.


\bibliography{mybib}{}
\bibliographystyle{amsplain}

\end{document}